\documentclass[10pt,final]{IEEEtran}
\usepackage{amsmath}
\usepackage{amsthm}
\newtheorem{theorem}{Theorem}
\usepackage{amssymb}
\usepackage{amsfonts}
\usepackage{latexsym}
\usepackage{graphicx}
\usepackage{epstopdf}
\usepackage{enumerate}
\usepackage{multicol}
\usepackage{color}
\usepackage{array}
\usepackage{algorithm,algorithmic}
\usepackage{bm}
\IEEEoverridecommandlockouts
\allowdisplaybreaks[4]

\begin{document}
\title{Metasurface Antenna-Enabled LEO Satellite Constellation Communications: Design and Optimization}
\author{Wenfei Yao, Xiaoming Chen, Qi Wang, Qiao Qi, and Ming Ying
\thanks{Wenfei Yao, Xiaoming Chen, Qi Wang, and Ming Ying are with the College of Information Science and Electronic Engineering, Zhejiang University, Hangzhou 310027, China (e-mails: \{yaowenfei, chen\underline{~}xiaoming, wang-qi, ming\underline{~}ying\}@zju.edu.cn). Qiao Qi is with the School of Information Science and Technology, Hangzhou Normal University, Hangzhou 311121, China (e-mails: qiqiao@hznu.edu.cn).}}
\maketitle

\begin{abstract}
Next-generation low Earth orbit (LEO) satellite constellations face critical bottlenecks in spectral efficiency and onboard hardware complexity.  To overcome these limitations, this paper introduces a novel architecture enabled by metasurface antennas (MAs) at the LEO satellites. In particular, MAs are metasurface-integrated feed antennas that perform high-precision beamforming directly in the wave domain, thereby effectively mitigating multi-user interference. Based on such an antenna architecture, a weighted sum rate (WSR) maximization problem is formulated by jointly optimizing the scheduling of feed antennas to terrestrial users (TUs) and the passive beamforming of the metasurface for system performance enhancement. To address this mixed-integer nonlinear programming (MINLP) challenge, an alternating optimization (AO)-based joint scheduling and beamforming algorithm is proposed. On the one hand, the proposed algorithm incorporates a polynomial-time minimum-cost maximum-flow (MCMF) method, which is dedicated to the optimal scheduling of feed antennas and TUs. On the other hand, it adopts a weighted minimum mean square error (WMMSE) method integrated with semidefinite relaxation (SDR) technique, which is tailored for metasurface beamforming design. Simulation results confirm the effectiveness of the proposed algorithm for MA-enabled LEO satellite constellation communications.
\end{abstract}

\begin{IEEEkeywords}
6G, metasurface, LEO satellite constellation, user scheduling, beamforming design
\end{IEEEkeywords}

\section{Introduction}
Satellite communication (SATCOM) has emerged as a pivotal part in the envisioned sixth-generation wireless networks, providing the essential infrastructure for ubiquitous connectivity. Among the satellites operate in different types of orbits, low Earth orbit (LEO) satellites have lower costs, shorter transmission delays, and higher signal power \cite{LEO satellite advantage}. Driven by the significant reduction in rocket launch costs and breakthroughs in communication technologies, LEO satellite constellations have entered a phase of prosperous development, with industry leaders including Starlink, OneWeb, Kuiper, and Telesat making significant progress in deploying their respective networks \cite{LEO satellite constellation}. By leveraging densely interconnected satellite networks, the LEO constellations are poised to achieve truly seamless global coverage and robust, high-performance interconnection services \cite{LEO satellite constellation project}.

To satisfy the escalating requirements for both increased data traffic and higher quality of service (QoS), multi-beam satellite system (MBSS) has been extensively studied in SATCOMs \cite{MSS}--\cite{MSS2}. A key challenge in MBSS operation is severe inter-beam interference, which acts as a major factor degrading system throughput and QoS performance. To mitigate this issue, digital beamforming techniques are widely adopted in MBSSs. Moreover, multi-satellite cooperation with distributed beamforming have been studied for LEO satellite constellations to mitigate significant co-channel interference \cite{Distributed beamforming}. This approach enables adjacent satellites to share channel state information (CSI) and collaboratively serve terrestrial users (TUs) through cooperative beamforming and dynamic user scheduling.  This hierarchical coordination paradigm is conceptually aligned with software-defined networking principles, which enable scalable and energy-efficient resource management in complex networks, including emerging 5G/6G multimedia internet of vehicles (IoV) systems \cite{Montazerolghaem2025IoV}. For example, the authors in \cite{leader-follower strcuture} introduced a hierarchical coordination framework for satellite clusters, where a designated leader satellite orchestrates cluster operations through inter-satellite links (ISLs), including information exchange and control signaling with follower satellites.  The work \cite{User scheduling of cooperative satellite} proposed a computationally efficient cooperative framework for LEO constellation networks, achieving improved spectral efficiency via coordinated beamforming and optimized user scheduling. Moreover, the study \cite{Link Scheduling in Satellite Networks} investigated dynamic link scheduling in LEO satellite networks, accounting for spatio-temporal correlations arising from orbital motion patterns.

The above studies have all proven the effectiveness of MBSS in satellite constellations. Nevertheless, implementing digital precoding necessitates deploying high performance baseband processing units onboard the satellite, leading to elevated power demands \cite{Digital beamforming}. In particular, MBSS relys on digital beamforming with high-resolution digital-to-analog converters and radio frequency (RF) chains, incurring prohibitive hardware expenses. While hybrid beamforming techniques offer a potential reduction in hardware overhead, the constrained waveform design flexibility may compromise system performance. Fortunately, metasurface, also known as reconfigurable intelligent surface (RIS), have demonstrated remarkable potential for both spectral and energy efficiency enhancement in next-generation wireless systems. Specifically, metasurface is a two-dimensional planar array that can employ programmable meta-atom elements to precisely control wave propagation characteristics with low hardware overhead \cite{RIS definiation}. Since they operate without conventional RF amplification chains and mainly rely on low-power tuning elements to adjust their phase responses, their power consumption and heat dissipation are typically much lower than those of fully digital or hybrid beamforming architectures employing multiple active RF chains \cite{YWF1, YWF2}. Existing studies indicate that the overall power consumption and thermal footprint of metasurface-based systems are dominated by lightweight control and biasing circuitry, which remain modest even for large surface sizes \cite{RIS_low_energy}. Typically, metasurface is employed to intelligently circumvent signal blockages by dynamically redirecting incident waves toward intended receivers. This inherent capabilities of metasurface make them especially valuable in SATCOM, where they can mitigate propagation loss and improve link quality by precisely directing signals to TUs \cite{SATCOM RIS}.

Furthermore, emerging metasurface antenna (MA)-enabled architectures demonstrate improved energy efficiency and hardware efficiency, positioning them as viable successors to digital beamforming technologies in future communication systems. In particular, a MA generally consists of a feed antenna array structure coupled with a metasurface layer composed of meta-atoms embedded with active tuning components, enabling software-defined beamforming and adaptive radiation patterns \cite{MA1}, \cite{MA2}. By joint optimizing scheduling of feed antenna and user and metasurface phase shifts, the MA enables native wave-domain signal processing through precise meta-atom transmission coefficient adjustments \cite{Optimization of MA}. This architecture effectively circumvents conventional digital signal processing chains, offering a more efficient and adaptive beamforming solution \cite{RIS-based antenna 1}. The distinctive beamforming and wavefrom manipulation abilities of MA have motivated numerous studies investigating their practical implementations. In \cite{RIS-based antenna 2}, Dai et al. presented the first wireless communication prototype employing a 256-element RIS, demonstrating significant power savings over phased arrays while achieving comparable gain. Experiments at 2.3 GHz and 28.5 GHz validated its real-time beamforming capability for high-definition video transmission. The authors in \cite{RIS channel model} derived physics-based free-space path-loss models for RIS-assisted links and verified them with anechoic-chamber measurements. The work \cite{RIS-integrated BS} proposed a RIS-integrated base station that deploys RIS close to the feed antenna of base station, enabling low-complexity, passive beamforming for uplink MIMO. A beam-search plus path-antenna pairing algorithm is devised and experimentally verified to maximize sum-rate while avoiding per-element channel estimation. In addition, the authers in \cite{SIM Satellite 1} demonstrated that mounting MA on LEO satellites enabled wave-domain multiuser beamforming, replacing digital precoding and achieving near-optimal sum rate performance using only statistical CSI. The work in \cite{SIM Satellite 2} presented a MA-equipped LEO satellite and evaluated its path-loss performance for extreme satellite-to-ground links under dynamic elevation angles. It is worth noting that practical reconfigurable intelligent surface (RIS) or metasurface-assisted wireless systems are inevitably affected by hardware impairments, such as phase noise, finite phase resolution, nonlinear response of tuning elements, and channel estimation errors. These non-idealities may degrade the achievable performance if not properly accounted for. Recent studies \cite{Li_TVT_2024}-\cite{Liu_GLOBECOM_2020} have investigated the impact of hardware impairments on RIS-aided systems and proposed robust transmission and channel estimation techniques to mitigate their effects.

Although existing studies have explored MA in SATCOM, they are predominantly limited to single-satellite scenarios \cite{SIM Satellite 1}, \cite{SIM Satellite 2}. Furthermore, much of this existing research \cite{Optimization of MA}-\cite{SIM Satellite 2} adopts a decoupled approach to system design, focusing solely on addressing metasurface beamforming in isolation while overlooking the vital need for joint optimization of user scheduling and feed antenna allocation. Yet scheduling plays a critical role in maximizing the value of beamforming. A well-designed scheduling strategy can intelligently match feed antennas to TUs based on channel conditions, ensuring the enhanced signal gains from metasurface beamforming are effectively allocated to TUs. This not only avoids wasting beamforming gains on poorly matched feed antenna-TU pairs but also further boosts overall system performance. To the best of our knowledge, a holistic framework that integrates these coupled design aspects for the LEO satellite constellation system remains unexplored. In this context, this paper addresses this critical research gap by proposing a MA-enabled LEO satellite constellation communication system, with the major contributions being summarized as follows:
\begin{enumerate}

    \item We put forward a MA-enabled LEO satellite constellation communication system that leverages inter-satellite cooperation and employs metasurfaces to perform low-cost beamforming in the wave-domain, effectively suppressing multi-user interference.

    \item We formulate a mixed-integer nonlinear programming (MINLP) problem in order to maximize the system's weighted sum rate (WSR). Then, an alternating optimization (AO)-based algorithm is proposed, wherein the scheduling of feed antenna and TUs is addressed by a minimum-cost maximum-flow (MCMF) method, while the metasurface beamforming is optimized via the weighted minimum mean square error (WMMSE) method integrated with semidefinite relaxation (SDR) technique.
	
	\item We explore the convergence and complexity properties of the proposed algorithm. Extensive numerical experiments further validate the effectiveness of the proposed algorithm, demonstrating substantial performance improvements over existing benchmarks.
\end{enumerate}

\emph{Organization:} The subsequent sections of this paper are structured as follows. Section II introduces the system model of MA-enabled LEO satellite constellation communication and formulates the optimization problem. Then, Section III proposes a joint scheduling and beamforming design algorithm for MA-enabled LEO satellite constellation communication. Section IV validates the effectiveness of the proposed algorithm through simulation results. Lastly, Section V concludes the paper.

\emph{Notations}:  Bold lower-case and upper-case letters denote column vectors and matrices, respectively. $\mathbb{C}^{M\times N}$ denotes a complex matrix of size $M \times N$. $\odot$ represents the Hadamard product. $\otimes$ denotes the Kronecker product. $(\cdot)^{*}$ denotes the complex conjugate operation. $(\cdot)^{\text{T}}$ and $(\cdot)^\text{H}$ denote the transpose and conjugate transpose, respectively.  $\mathcal{CN}(\mu,\sigma^2)$ denotes the circularly symmetric complex Gaussian distribution with $\mu$ being the mean and $\sigma^2$ being the variance. $|\cdot|$ and $\|\cdot\|$ denote the absolute value of scalar and the L2-norm of vector, respectively. $\text{diag}(\cdot)$ and $\text{vec}(\cdot)$ denote the process of diagonalization and vectorization, respectively. $\left[\mathbf{A} \right]_{l,l}$ denotes the $l$-th diagonal elements of the matrix $\mathbf{A}$.

\section{System Model}
\begin{figure}[!ht]
	\centering
	\includegraphics[width=3.4in]{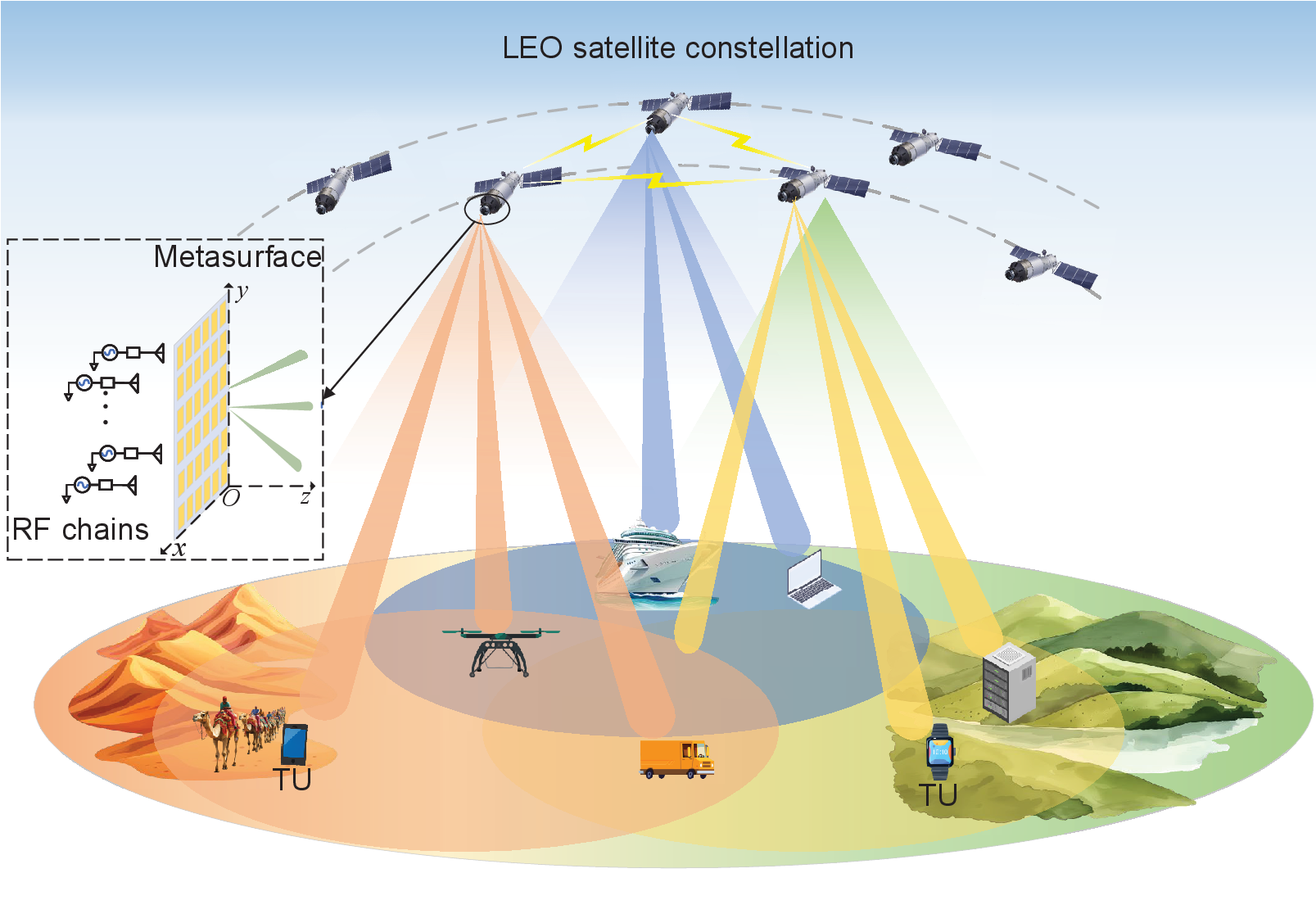}
	\caption{System model of the MA-enabled LEO satellite constellation communication}
	\label{system model}
\end{figure}
Consider a MA-enabled LEO satellite constellation communication system where $S$ satellites provide communication services to $K$ single-antenna TUs, as shown in Fig. \ref{system model}. Let $\mathcal{U}_{s}$ denote the set of TUs located within the coverage area of the $s$-th satellite and let $\mathcal{S}_{k}$ denote the set of visible satellites for the $k$-th TU. Each LEO satellite is equipped with an autonomous control MA and adopts full frequency multiplexing to maximize frequency utilization. Specifically, the MA consists of $N_t$ feed antennas each powered by an RF chain and a metasurface composed of a uniform planar array (UPA) with $L = L_x \times L_y $ meta-atom elements, where $ L_x$  and $ L_y $ denote the number of elements along the column and row direction, respectively. In the considered MA-enabled LEO satellite constellation communication system, each TU is solely served by one feed antenna of a LEO satellite selected from its visible satellite set $\mathcal{S}_k$. Specifically, the signal stream required by a TU is transmitted exclusively through its assigned feed antenna, whereupon the metasurface dynamically shapes the signal. In particular, a leader satellite is designated, and the remaining $S-1$ satellites connect to it via ISLs to share real-time information and necessary CSI \footnote{ It is noted that the coordination between the leader and follower satellites relies on high-capacity ISLs, e.g., optical or millimeter-wave/terahertz ISLs, over which CSI and optimized parameters are exchanged. Therefore, the associated signaling overhead and latency are negligible relative to the optimization time scale.}. Additionally, the leader satellite leverages its onboard processing capabilities to support the dynamic optimization of feed antenna-TU scheduling and metasurface beamforming parameters.

\subsection{Channel Model}
In the signal transmission process, the signal emitted by each feed antenna first propagates through a near-field channel to the metasurface, undergoes wave-domain beamforming via metasurface, and is then transmitted to TUs through the satellite-terrestrial channel. Accounting for these cascaded propagation and beamforming stages, the equivalent channel coefficient from the $n$-th feed antenna of the $s$-th satellite to the $k$-th TU is formulated as
\begin{align}
	h^{\text{equ}}_{s,n,k}(t) = \mathbf{h}^{\text{H}}_{s,k}(t)\boldsymbol{\Theta}_{s}\mathbf{g}_{s,n} \in \mathbb{C}^{1\times 1},
\end{align}
where $\mathbf{g}_{s,n}\in \mathbb{C}^{L \times 1}$ represents the near-field channel vector between the $n$-th feed antenna and the metasurface of the $s$-th satellite, $\mathbf{h}_{s,k}(t)\in \mathbb{C}^{L \times 1}$ denotes the satellite-terrestrial channel vector between the metasurface of the $s$-th satellite and the $k$-th TU, and $\boldsymbol{\Theta}_{s} \in \mathbb{C}^{L \times L}$ denotes the beamforming matrix of the metasurface of the $s$-th satellite. In particular, $\boldsymbol{\Theta}_{s} = \text{diag}\{\boldsymbol{\varphi}_{s}\}$, where $\boldsymbol{\varphi}_{s} = [e^{j\theta_{s,l}},e^{j\theta_{s,2}},\cdots,e^{j\theta_{s,L}}] \in \mathbb{C}^{1 \times L}$ with $\theta_{s,l}$ denoting the adjustable phase shift of the $l$-th element of the $s$-th satellite. It is noted that the metasurface considered in this work operates in a phase-only reconfiguration mode, where each meta-atom adjusts only the phase of the impinging signal without introducing active RF amplification. As a result, the metasurface does not generate additional transmit power, and its energy consumption is primarily associated with low-power control circuitry used for phase tuning \cite{RIS_energy1, RIS_energy2}.

In general, for the near-field channel vector $\mathbf{g}_{s,n}$, it aggregates the transmission coefficients from the $n$-th feed antenna to each of the $L$ meta-atoms on metasurface, which can be formulated as $\mathbf{g}_{s,n} = [g^{1}_{s,n}; \cdots; g^{L}_{s,n}] \in \mathbb{C}^{L \times 1}$ with $g^{l}_{s,n}$ denoting the element-wise transmission coefficient. According to the physical propagation characteristics of the near-field channel, $g^{l}_{s,n}$ can be modeled based on the Rayleigh-Sommerfeld diffraction theory \cite{Rayleigh-Sommerfeld diffraction theory}. Specifically, $g^{l}_{s,n}$ is given by
\begin{align}
 g^{l}_{s,n} = \frac{A_{t}\cos \tau^{l}_{s,n}}{r^{l}_{s,n}}\left(\frac{1}{2\pi r_{s,n}^{l} }- j\frac{1}{\lambda}\right)e^{j2\pi r_{s,n}^{l}/\lambda},
\end{align}
where $A_{t}$ denotes the area of each meta-atom, $r_{s,n}^{l}$ denotes the Euclidean distance of signal propagation, $\tau_{s,n}^{l}$ denotes the angle bewteen the signal propagation direction and the normal direction of the metasurface, and $\lambda$ denotes the wave-length\footnote{ It should be noted that realistic hardware implementations may introduce discrepancies in the transmission parameters between feed antennas and the metasurface due to manufacturing tolerances. Nevertheless, such deviations can be effectively compensated through calibration procedures employing conventional error back-propagation techniques \cite{hardware imperfect}. Therefore, the near-field channel between the feed antennas and the metasurface is assumed to be perfectly known in this work, serving as a deterministic and calibrated system component for the subsequent communication-level optimization.}.

Moreover, for the satellite-terrestrial channel vector $\mathbf{h}_{s,k}(t)$ at time slot $t$, it is formulated as
\begin{align}
	\mathbf{h}_{s,k}(t) = &\sum_{p=0}^{N_{s,k}-1}\eta_{s,k,p}\cdot \exp\{j2\pi(t\nu_{s,k,p}-f_{c}\tau_{s,k,p})\} \notag \\ &\cdot \mathbf{a} ( \theta_{s,k,p},\psi_{s,k,p}),
\end{align}
where $N_{s,k}$ represents the number of paths from the $s$-th satellite to the $k$-th TU, $f_c$ represents the carrier frequency, and $\eta_{s,k,p}$, $\nu_{s,k,p}$, and $\tau_{s,k,p}$ represent the complex channel gain, Doppler shift, and propagation delay with respect to the $p$-th path, respectively. Moreover, $\mathbf{a}(\theta_{s,k,p},\psi_{s,k,p})$ represents the UPA response of the $p$-th path of the $s$-th satellite's metasurface. The satellite-terrestrial channel consists of a dominant line of sight (LoS) component (index $p=0$) and multiple non-line of sight (NLoS) components (index $p=1,\cdots, N_{s,k}-1$). The NLoS component consists of multiple scattered paths originating from ground reflections and local scatterers. Due to the large altitude of the satellite relative to terrestrial objects, these multipath components are observed within a narrow angular spread at the satellite side. Consequently, their angles of arrival are modeled as small perturbations around the nominal LoS direction. Specifically, the array response of each NLoS path can be expressed as
$\mathbf{a}(\theta_{s,k,0} + \delta_{s,k,p}, \psi_{s,k,0} + \xi_{s,k,p})$, where $\delta_{s,k,p}$ and $\xi_{s,k,p}$ denote small angular deviations. In particular, the UPA response vertor $\mathbf{a}(\theta,\psi)$ can be formulated as
\begin{align}
	\mathbf{a}(\theta,\psi) =&\frac{1}{\sqrt{L_x}} [1;e^{j\varpi\cos(\theta)\sin(\psi)};\cdots;e^{j\varpi(L_x - 1)\cos(\theta)\sin(\psi)}]\notag \\ & \otimes \frac{1}{\sqrt{L_y}}[1;e^{j\varpi\cos(\psi)};\cdots;  e^{j\varpi(L_y - 1)\cos(\psi)}],
\end{align}
where $\varpi=2\pi f_c \frac{d}{c}$ with $c$ being the light velocity and $d$ being the inter-element spacing of metasurface.

Due to fast mobility, LEO satellite networks experience significantly larger Doppler shifts. For the Doppler shift $\nu_{s,k,p}$, it includes two independent components, i.e., $\nu_{s,k,p} = \nu_{s,k,p}^{\text{sat}}+\nu_{s,k,p}^{\text{tu}}$, where $\nu_{s,k,p}^{\text{sat}}$ and $\nu_{s,k,p}^{\text{tu}}$ denote the Doppler shift caused by the motion of satellite and TU, respectively. Considering the high altitude of the satellite orbit, the Doppler shift $\nu_{s,k,p}^{\text{sat}}$ caused by satellite motion are assumed to be path-independent, i.e., $\nu_{s,k,p}^{\text{sat}} = \nu_{s,k}^{\text{sat}}, \forall p$. Whereas, the component $\nu_{s,k,p}^{\text{tu}}$ arising from motion of the TU typically varies in different propagation paths and are much smaller than that brought by satellite motion, contributing to the Doppler spread $\Delta \nu = \max_{p}\nu_{s,k,p}^{\text{tu}} - \min_{p}\nu_{s,k,p}^{\text{tu}}$. In addition, the propagation delay bewteen the $s$-th LEO satellite and the $k$-th TU relative to the $l$-th path mainly composed of two components, i.e., $\tau_{s,k,p} = \tau^{\text{sat}}_{s,k} + \tau^{\text{tu}}_{s,k,p}$, where $\tau^{\text{sat}}_{s,k}$ is the deterministic free-space delay component, determined by the line-of-sight (LoS) distance between the $s$-th satellite and $k$-th TU. The second component $\tau^{\text{tu}}_{s,k,p}$ represents the stochastic scattering-induced delay caused by local multipath around the TU, which does contribute to delay spread effect through $\Delta \tau = \max_{p}\tau_{s,k,p}^{\text{tu}} - \min_{p}\tau_{s,k,p}^{\text{tu}}$. Since the delay spread $\Delta \tau$ is significantly smaller than the bulk delay, as demonstrated in \cite{Delay spread 1} and \cite{Delay spread 2}, we assume the delay remains constant across all paths and omit the path index $p$, i.e., $\tau_{s,k,p} = \tau_{s,k}, \forall p$ \cite{Doppler spread and Delay spread}.

\subsection{Signal Model}
The MA architecture avoids the use of digital precoding, resulting in each feed antenna serving a single TU with one dedicated data stream. This operational constraint necessitates careful scheduling between feed antennas and TUs to maximize the system performance. For the scheduling bewteen TUs and feed antennas of LEO satellites, we introduce $\alpha_{s,n,k}$ to represent the scheduling bewteen the feed antenna and TU, where $\alpha_{s,n,k} = 1$ denotes the $k$-th TU scheduled with the $n$-th feed antenna of the $s$-th satellite, and $\alpha_{s,n,k} = 0$ otherwise.

In this case, the time-domain signal transmitted by the $n$-th feed antenna of the $s$-th satellite can be formulated as
\begin{align}
	e_{s,n}(t) = p_{s,n}\sum_{k \in \mathcal{U}_{s}} \alpha_{s,n,k}x_{k}(t),
\end{align}
where $p_{s,n}$ represents the transmit power allocated to the $n$-th feed antenna of the $s$-th satellite, $x_{k}$ represents the data symbol intended for the $k$-th TU, which is assumed to be independent and satisfies $\mathbb{E}(\vert x_{k}x_{k}^{*}\vert) = 1$ and $\mathbb{E}(\vert x_{k}x_{k'}^{*}\vert) = 0, k\neq k'$. Due to the full frequency multiplexing, each TU experiences inter-satellite interference originating from all visible satellites within its LoS visible range\footnote{The interference from non-visible satellites of TUs is considered negligible due to combined effects of strong NLOS signal attenuation, inherent beamforming isolation, and frequency reuse constraints in LEO satellite constellations \cite{NLoS interference 1}, \cite{NLOS interference 2}.}. Under this context, the singal received by the $k$-th TU can be represented as
\begin{align}
	y_{k}(t) = &\sum_{s\in\mathcal{S}_{k}}\sum_{n=1}^{N_t} h^{\text{equ}}_{s,n,k}(t) p_{s,n}\sum_{i \in \mathcal{U}_{s}} \alpha_{s,n,i}x_{i}(t) + n_{k}(t)\notag\\
	= &\underbrace{\sum_{s\in \mathcal{S}_{k}}\sum_{n=1}^{N_t}\alpha_{s,n,k}p_{s,n}h^{\text{equ}}_{s,n,k}(t)x_{k}(t)}_{\text{desired signal}} \notag \\&+ \underbrace{\sum_{s\in \mathcal{S}_{k}}\sum_{n=1}^{N_t}\sum_{k'\neq k}\alpha_{s,n,k'}p_{s,n}h^{\text{equ}}_{s,n,k}(t)x_{k'}(t)}_{\text{interfrence}} \notag \\&+ \underbrace{n_{k}(t)}_{\text{AWGN}},
\end{align}
where $n_{k}$ is additive white Gaussian noise (AWGN) with zero mean and variance $\sigma_{k}^2$. In particular, $\sigma_{k}^2 = B \zeta \kappa$ with $B$, $\zeta$, and $\kappa$ denoting the channel bandwidth, noise temperature and Boltzmann constant, respectively. Given the characteristics of satellite-terrestrial channels under high-speed movement, Doppler shift and propagation delay compensation are implemented at the TU side. The Doppler shift and propagation delay are derived from the satellite and TU positions, which is obtained through satellite ephemeris data and global navigation satellite system measurements, respectively \cite{Doppler shift estimation}. Since the channel is dominanted by the LoS link, the received signal is compensated using the LoS path-specific Doppler shift and time delay. For the $k$-th TU, considering its Doppler shift and propagation delay compensation coefficients relative to the scheduling $\alpha_{s,n,k}$, the compensation coefficients are given by
\begin{align}
	\nu_{k}^{\text{com}} &= \sum\limits_{s\in\mathcal{S}_{k}}\sum\limits_{n=1}^{N_t}\alpha_{s,n,k}(\nu_{s,k}^{\text{sat}}+\nu_{s,k,1}^{\text{tu}}),
\end{align}
and
\begin{align}
	\tau_{k}^{\text{com}} &= \sum\limits_{s\in\mathcal{S}_{k}}\sum\limits_{n=1}^{N_t}\alpha_{s,n,k}\tau_{s,k}.
\end{align}
In this context, the compensated signal can be formulated as
\begin{align}
	y^{\text{com}}_{k}(t) =& y_{k}(t + \tau^{\text{com}}_{k}) \cdot \text{exp}\{-j2\pi(t+\tau^{\text{com}}_{k})\nu_{k}^{\text{com}}\}\notag\\
	= &\sum_{s\in\mathcal{S}_{k}}\sum_{n=1}^{N_t}\alpha_{s,n,k}p_{s,n}\hat{h}_{s,n,k}^{\text{equ}}(t) x_{k}(t)\notag\\
	& + \sum_{s\in \mathcal{S}_{k}}\sum_{n=1}^{N_t}\sum_{k'\neq k}\alpha_{s,n,k'}p_{s,n} \hat{h}_{s,n,k}^{\text{equ}}(t) x_{k'}(t) \notag \\&+ n_{k}(t),
\end{align}
where $\hat{h}^{\text{equ}}_{s,n,k}(t) = \hat{\mathbf{h}}^{\text{H}}_{s,k}(t)\boldsymbol{\Theta}\mathbf{g}_{s,n}$ represents the compensated equivalent channel. In particular, $\hat{\mathbf{h}}_{s,k}(t)$ is the compensated satellite-terrestrial channel, which is given by
\begin{align}
	\hat{\mathbf{h}}_{s,k}(t) = &\sum_{p=0}^{N_{s,k}-1}\hat{\eta}_{s,k,p} \cdot
	\exp\{j2\pi(t\nu_{s,k,p}^{\text{com}} - f_{c} \tau_{s,k}^{\text{com}})\} \notag \\ &\cdot \mathbf{a} ( \theta_{s,k},\psi_{s,k}),
\end{align}
where $\hat{\eta}_{s,k,p} = \eta_{s,k,p}\cdot \exp\{j2\pi\tau_{k}^{\text{com}}\nu_{s,k,p}^{\text{tu}}\}$,  $\nu_{s,k,p}^{\text{com}} = \nu_{s,k,p} - \nu_{k}^{\text{com}}$ and $\tau_{s,k}^{\text{com}} = \tau_{s,k} - \tau_{k}^{\text{com}}$ represents the compensated frequency shift and propagation delay, respectively.

In practical LEO SATCOM systems, perfect CSI is difficult to obtain due to channel estimation inaccuracies, residual scattering dynamics, and hardware imperfections. In this work, Doppler shift and propagation delay are assumed to be deterministically compensated using satellite ephemeris data and global navigation satellite system (GNSS) measurements prior to channel estimation. As a result, the remaining CSI uncertainty primarily affects the Doppler/delay compensated channel. Accordingly, the compensated satellite--terrestrial channel is modeled as
\begin{equation}
	\hat{\mathbf{h}}_{s,k}(t) = \tilde{\mathbf{h}}_{s,k}(t) + \mathbf{e}_{s,k}(t),
	\label{eq:channel_error_model}
\end{equation}
where $\tilde{\mathbf{h}}_{s,k}(t) \in \mathbb{C}^{L \times 1}$ denotes the estimated compensated channel available at the satellite, and $\mathbf{e}_{s,k}(t)$ represents the residual channel estimation error. The error vector is modeled as a circularly symmetric complex Gaussian random vector
\begin{equation}
	\mathbf{e}_{s,k}(t) \sim \mathcal{CN}\left(\mathbf{0}, \sigma_{e}^{2} \mathbf{I}\right),
\label{eq:error_distribution}
\end{equation}
where $\sigma_{e}^{2}$ characterizes the error variance. The channel error variance is modeled as inversely proportional to the singal to noise ratio (SNR), i.e., $\sigma_e^2 = 1/\mathrm{SNR}$, which captures the dependence of estimation accuracy on the received signal quality \cite{sigmae}. Under this model, the equivalent cascaded channel between the $n$-th feed antenna of the $s$-th satellite and the $k$-th TU becomes
\begin{equation}
	\hat{h}^{\mathrm{equ}}_{s,n,k}(t) = \left(\tilde{\mathbf{h}}_{s,k}(t) + \mathbf{e}_{s,k}		(t)\right)^{H} \Theta_s \mathbf{g}_{s,n}.
	\label{eq:equ_channel_error}
\end{equation}

For the link from the scheduled $s$-th satellite and the $k$-th TU, i.e., $\alpha_{s,n,k} = 1$, the residenual Doppler shift $\nu_{s,k,p}^{\text{com}} = \nu_{s,k,p}^{\text{tu}}$ and delay $\tau_{s,k}^{\text{com}} = \tau_{s,k}^\text{tu}$ are negligible and are acceptable for the reliable transmission \cite{Residual Doppler}. Nevertheless, the signal from the other satellites are still asynchronous at the $k$-th TU with a significant Doppler shift and delay. In addition, the compensated channel is affected by estimation errors, which introduces uncertainty in the received signal power. Under this context, the signal to interference plus noise ratio (SINR) at the $k$-th TU is characterized in an average sense and can be
expressed as~\cite{SINR with asynchronous}
\begin{align}
    \mathrm{SINR}_k = \frac{ \sum\limits_{s \in \mathcal{S}_k} \sum\limits_{n=1}^{N_t} \alpha_{s,n,k} p_{s,n} \mathbb{E}\!\left[ \left| \hat{h}^{\mathrm{equ}}_{s,n,k} \right|^2 \right]}{\sum\limits_{s \in \mathcal{S}_k} \sum\limits_{n=1}^{N_t} \sum\limits_{k' \neq k} \alpha_{s,n,k'} p_{s,n} \rho_{s,k} \mathbb{E}\!\left[\left| \hat{h}^{\mathrm{equ}}_{s,n,k} \right|^2 \right] + \sigma_k^2 }, \label{SINR}
\end{align}
where
\begin{align}
    \mathbb{E} \left[\left| \hat{h}^{\mathrm{equ}}_{s,n,k} \right|^2 \right] = \left|       \tilde{\mathbf{h}}_{s,k}^H \Theta_s \mathbf{g}_{s,n} \right|^2 + \sigma_e^2 \left\| \Theta_s    \mathbf{g}_{s,n} \right\|^2
\end{align}
and $\rho_{s,k}$ represents the asynchronous interference signal power coefficient of the $s$-th satellite to the $k$-th TU, which is formulated as \cite{asynchronous factor}
\begin{align}\label{singal power coefficient}
	\rho_{s,k} = \left(1-\frac{\beta}{4} + \frac{\beta}{4}\cdot\cos(2\pi\tau_{s,k}^{\text{com}})\right),
\end{align}
where $\beta$ represents the roll-off factor of the normalized root raised consine function. Thus, the achievable rate $R_k$ evaluated based on the expected SINR under channel estimation errors, which is given by
\begin{align}\label{achievable rate}
	R_{k} = \log_2({1+\text{SINR}_{k}}).
\end{align}

It is evident that the $k$-th TU's achievable rate is jointly impacted by the scheduling and metasurface beamforming from the expression of SINR (\ref{SINR}). Specifically, the scheduling determines which satellite feed antennas serve each TU, it optimizes the use of strong combined channel gains, thereby boosting the desired signal power. In addition, metasurface directly shapes equivalent channel \(\hat{h}^{\text{equ}}_{s,n,k}(t,f_c)\). By concentrating desired signals onto target TU to enhance signal power and attenuating channel gain for non-target TU to suppress interference, the SINR of TUs are elevated. Given this coordinated effect, joint design of scheduling and beamforming is anticipated to effectively enhance the LEO satellite constellation communication quality.

\subsection{Problem Formulation}
In this paper, our target is to maximize the WSR of the TUs by optimizing the beamforming of metasurface and scheduling bewteen the feed antenna and TUs. Therefore, the optimization problem is mathematically formulated as
\begin{subequations}\label{OP0}
	\begin{align}
		\text{(P1):}~ \max_{\boldsymbol{\Theta}, \boldsymbol{\alpha}} ~&\sum_{k=1}^{K}\beta_{k}R_{k} \label{OP0_obj} \\
		\text{s.t.}
		~&\vert \left[\boldsymbol{\Theta}_{s,s} \right]_{l,l} \vert = 1, \forall s, \forall l,\label{OP0_c1}\\
		~&\sum_{k}\alpha_{s,n,k} \leq 1, \forall s,\forall n,\label{OP0_c2}  \\
		~&\sum_{s\in\mathcal{S}_{k}}\sum_{n=1}^{N_t}\alpha_{s,n,k} = 1, \forall k,\label{OP0_c3} \\
		~&\alpha_{s,n,k} \in \{0, 1\},\label{OP0_c4}
	\end{align}
\end{subequations}
where $\beta_{k}\in [0,1]$ represents the weight coefficient of the $k$-th TU, $\boldsymbol{\alpha} = \{\alpha_{s,n,k}|\forall s, \forall  n, \forall k\}$ represents the set of the scheduling factor, and $\boldsymbol{\Theta}=\{\boldsymbol{\Theta_{s}}|\forall s\}$ represents the set of the metasurface beamforming. The objective function (\ref{OP0_obj}) denotes the WSR of the TUs. Constraint (\ref{OP0_c1}) denotes the modulus constraint on the reflection coefficients of the metasurfaces. Constraint (\ref{OP0_c2}) denotes that there is no more than one TU served by each feed antenna. Constraint (\ref{OP0_c3}) denotes that each TU must be served by a feed antenna belonging to one of its visible satellite $\mathcal{S}_{k}$. Constraint (\ref{OP0_c4}) denotes the 0-1 constraint of the scheduling factor. The problem (P1) constitutes a MINLP problem, which presents significant computational challenges due to the non-convex objective function and constraints. To this end, we develop an efficient AO algorithm that yields a feasible suboptimal solution, ultimately improving the communication quality.

\section{Joint Design of Scheduling and Beamforming}
In this section, we focus on the design of a joint scheduling and beamforming algorithm for MA-enabled LEO satellite constellation communications by solving the optimization problem (P1). Firstly, problem (P1) is decomposed into two subproblems, i.e., the scheduling of feed antennas and TUs subproblem and the metasurface beamforming optimization subproblem. Then, the subproblems are solved iteratively until convergence. In particular, the mixed-integer nonlinear scheduling subproblem is solved using a novel MCMF-based algorithm, while the beamforming optimization subproblem is handled through the SDR technique plus penalty-function. In the following, we introduce the detailed algorithms for these two subproblems, respectively.

\subsection{Scheduling of Feed Antenna and TUs}
With a fixed metasurface beamforming $\boldsymbol{\Theta}$, the scheduling of feed antenna and TUs subproblem can be expressed as
\begin{subequations}\label{OP1}
	\begin{align}
		\text{(P2):}~\max_{\boldsymbol{\alpha}} ~&\sum_{k=1}^{K}\beta_{k}R_{k}\label{OP1_obj} \\
		\text{s.t.} ~&\text{(\ref{OP0_c2})-(\ref{OP0_c4})}.
	\end{align}
\end{subequations}
The optimal scheduling between feed antennas and TUs constitutes a complex combinatorial optimization problem. While exhaustive search could theoretically evaluate all $A_{SN_t}^{K}$ possible assignments to identify the global optimum, the factorial growth of solution space renders this brute-force approach computationally intractable for practical system dimensions. Fortunately, it is worth noting that the scheduling problem inherently features a one-to-one matching constraint, where each feed antenna can serve at most one TU at a time, and each TU requires only one feed antenna to meet its communication needs. This constraint aligns well with the core logic of the MCMF problem, which excels at allocating flow between two node sets under strict capacity limits. Further, the scheduling objective of maximizing the WSR can be seamlessly adapted to MCMF by mapping the rate contribution of each feed antenna-TU pair to the negative cost of the corresponding edge in the flow network. Therefore, we develop an efficient polynomial-time solution based on the MCMF algorithm to address scheduling of feed antenna and TUs subproblem.

\subsubsection{MCMF Problem}
Given a directed flow network $G = (V, E, c, w)$ where $V$ represents the node set and $E$ denotes the arc set. Each directed edge $(v_{i},v_{j}) \in E$ has an associated capacity $c_{ij} \geq 0$ and cost coefficient $w_{ij}$. The MCMF problem aims to find a flow $f$ that maximizes the total flow while minimizing the transportation cost. The mathematical formulation is expressed as:
\begin{subequations}\label{MCMF}
	\begin{align}
		\min_{f}~&\sum_{(v_{i},v_{j})\in E}w_{ij}f(v_{i}, v_{j})\label{MCMF_obj}\\
		\text{s.t.} ~&0\leq f(v_i,v_j)\leq c_{ij},\label{MCMF_c1}\\
		&\sum_{v_{i}\in V-v_{j}}f(v_i,v_j) = \sum_{v_{j}\in E-v{_{i}}}f(v_{j},v_{i}),\label{MCMF_c2}\\
		&\vert f \vert = \sum_{v_{i}\in V-v_{s}}f(v_{s},v_{i}) = \sum_{v_{j}\in V-v_{t}}f(v_{j},v_{t}).\label{MCMF_c3}
	\end{align}
\end{subequations}
The above constraints describe the necessary conditions that the feasible flow $f$ must satisfy. The objective function (\ref{MCMF_obj}) represents the minimum cost of the flow $f$, where $f(v_{i}, v_{j})$ denotes the traffic passing through edge $(v_{i},v_{j})$. Constraint (\ref{MCMF_c1}) states that the flow must respect the capacity of the edge. Constraint (\ref{MCMF_c2}) represents the limit of zero traffic storage capacity for any intermediate node within the network. Constraint (\ref{MCMF_c3}) denotes that the total outgoing flow from the source $v_s$ equals the incoming flow at the sink $v_t$.

Based on this, we construct a flow network $G(U, V, E)$ to model the scheduling problem between feed antennas and TUs, as illustrated in Fig. 2. In this network, each TU is represented by a node $u_k \in U$ and each feed antenna is denoted by a node $v_j \in V$. The source node $v_s$ is connected to every TU node $u_k$ with an edge with capacity 1, while all antenna nodes $v_j$ are linked to the sink node $v_t$ through edges with capacity 1. For the bipartite connections, directed edges with capacity 1 are established between each TU node $u_k$ and its visible feed antenna nodes $v_j$ (where $v_j$ belongs to the satellites visible to the $k$-th TU). The resulting flow network directly corresponds to a standard maximum flow problem.
\begin{figure}[!ht]
	\centering
	\includegraphics[width=3.4in]{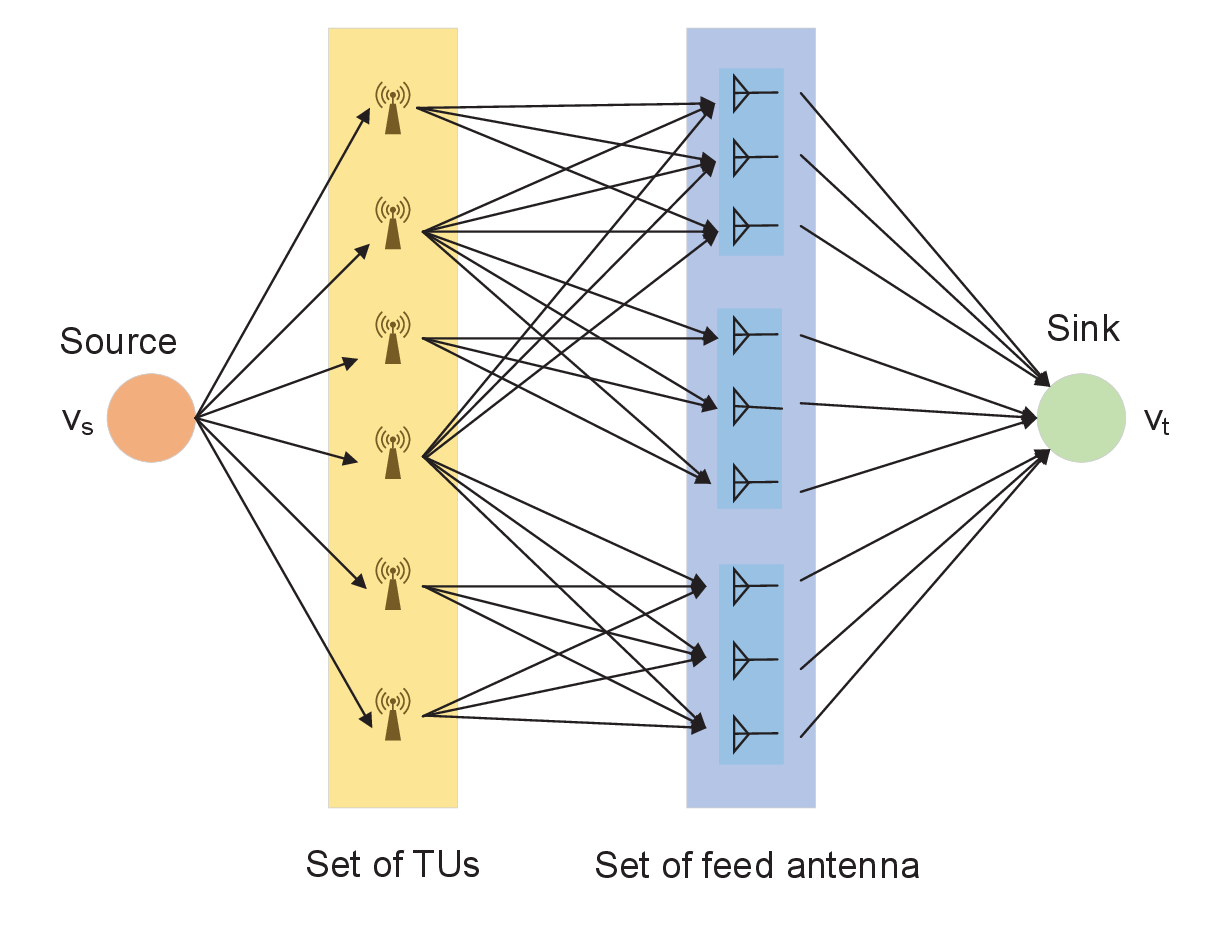}
	\caption{Constructed flow network for the scheduling problem}
	\label{Constructed flow}
\end{figure}

However, it is difficult to define the cost of the directed edges between the TU nodes and the feed antenna nodes and thus translate it into a MCMF problem. The inherent difficulty stems from the interdependent nature of the SINR expression of TUs, where the SINR of the $k$-th TU depends nonlinearly on both its serving feed antenna and interference from other active feed antennas. In other words, the set of active feed antennas itself is not pre-fixed, this means the sources and magnitudes of interference affecting the $k$-th TU's SINR cannot be determined independently of the overall feed antenna selection. To address this, we establish a theoretically rigorous lower bound on the achievable rate, corresponding to the worst-case that all other feed antennas are active and contribute to interference, which is given by
\begin{align}\label{lower_of_achievable_rate}
&R_{k} \geq \bar{R}_{k} = \log_2\left(1+\tilde{\Gamma}_{k}\right) \notag \\
&= \log_2\!\left( 1+ \frac{ \sum\limits_{s\in \mathcal{S}_{k}} \sum\limits_{n=1}^{N_t} \alpha_{s,n,k} p_{s,n} \mathbb{E}\!\left[ \left| \hat{h}^{\mathrm{equ}}_{s,n,k}(t) \right|^2 \right]
}{\sum\limits_{s\in \mathcal{S}_{k}} \sum\limits_{\substack{n=1 \\ \alpha_{s,n,k}\neq 1}}^{N_t} p_{s,n} \rho_{s,k} \mathbb{E}\!\left[ \left| \hat{h}^{\mathrm{equ}}_{s,n,k}(t) \right|^2 \right] + \sigma_k^2 } \right).
\end{align}
In (\ref{lower_of_achievable_rate}), the lower bound of achievable rate of the $k$-th TU is solely depends on the scheduling factor $\alpha_{s,n,k}$ of the $k$-th TU and $n$-th feed antenna. Thus, we define the edge cost between TU node $u_{k}$ and feed antenna node $v_{j}$ (with $j = (s-1)N_t+n$) as $-\beta_{k}\bar{R}_{j,k}$. Due to channel estimation errors in the compensated channel, the achievable rate is lower bounded using the expected SINR, which yields a deterministic metric for scheduling. Accordingly, the scheduling problem is reformulated as
\begin{subequations}\label{OP1_lower_bound}
	\begin{align}
		\text{(P2-1):}~\min_{\boldsymbol{\alpha}} ~&\sum_{k=1}^{K}-\beta_{k}\bar{R}_{k} \\
		\text{s.t.} ~&\text{(\ref{OP0_c2})-(\ref{OP0_c4})}.
	\end{align}
\end{subequations}
Eventually, by solving the MCMF problem in the constructed flow network, we obtain an optimal solution to the optimization problem (P2-1), which is proved in Theorem 1. The detailed of the MCMF-based scheduling algorithm is provided in Algorithm 1.

\begin{algorithm}
\renewcommand{\algorithmicrequire}{\textbf{Input:}}
\renewcommand{\algorithmicensure}{\textbf{Output:}}
\caption{MCMF-based Algorithm for Scheduling of feed antenna and TUs}
\label{alg:unified_mcmf}
\begin{algorithmic}[1]
    \REQUIRE
        Flow network $G = (U, V, E, c, w)$ with:
        \STATE \quad $U$: Set of TUs;
        \STATE \quad $V$: Set of feed antennas;
        \STATE \quad $c(u,v)$: Edge capacities;
        \STATE \quad $w(u,v)$: Edge costs ($-\beta_{k}\bar{R}_{l,k}$);
        \STATE \quad Source $v_s$, sink $v_t$;
    \ENSURE Optimal scheduling matrix $\boldsymbol{\alpha}$;

    \STATE \textbf{Initialize:}
    \STATE $\boldsymbol{\alpha} \gets \mathbf{0}_{|U|\times|V|}$, $f \gets \mathbf{0}_{|E|}$, $\text{total\_cost} \gets 0$

    \WHILE{true}
        \STATE \textbf{Step 1: Build Residual Network $G_f$:}
        \FOR{$(u,v) \in E$}
            \IF{$f(u,v) < c(u,v)$}
                \STATE $E_f \gets E_f \cup \{(u,v)\}$ with $w_f(u,v) = w(u,v)$;
            \ENDIF
            \IF{$f(u,v) > 0$}
                \STATE $E_f \gets E_f \cup \{(v,u)\}$ with $w_f(v,u) = -w(u,v)$;
            \ENDIF
        \ENDFOR

        \STATE \textbf{Step 2: Find Min-Cost Path via shortest path faster algorithm (SPFA):}
        \STATE $p, \text{dist} \gets \text{SPFA}(G_f, v_s, v_t)$; \COMMENT{Using shortest path faster algorithm (SPFA) \cite{SPFA}}
        \IF{$p = \emptyset$}
            \STATE Break; \COMMENT{No augmenting path exists.}
        \ENDIF

        \STATE \textbf{Step3: Compute Bottleneck Capacity on path $p$:}
        \STATE $\delta \gets \min\{c_f(u,v) | (u,v) \in p\}$;
        \STATE \textbf{Step 4: Augment Flow:}
        \FOR{$(u,v) \in p$}
            \STATE $f(u,v) \gets f(u,v) + \delta$;
            \STATE $f(v,u) \gets f(v,u) - \delta$;
            \STATE $\text{total\_cost} \gets \text{total\_cost} + \delta \cdot w_f(u,v)$;
        \ENDFOR
    \ENDWHILE

    \STATE \textbf{Step 5: Derive Assignment:}
    \FOR{$u_k \in U, v_j \in V$}
        \STATE $\alpha_{s,n,k} \gets \mathbb{I}[f(u_k,v_j) = 1]$, $j = (s-1)N_t+n$;
    \ENDFOR

    \RETURN $\boldsymbol{\alpha}$.
\end{algorithmic}
\end{algorithm}

\begin{theorem}
The MCMF-based scheduling algorithm produces an optimal scheduling of TUs to feed antennas that maximizes the lower bound of the WSR while satisfying all system constraints.
\end{theorem}

\begin{proof}
The proof consists of two parts:
\begin{enumerate}
    \item \textbf{Equivalence to MCMF Problem:}
    The integral capacities ensure $f(u_k,v_j) \in \{0,1\}$, yielding an integral solution. Since the capacity between the source node $v_{s}$ to the TU node $u_{k}$ is 1, we have $\sum_{j} f(u_k,v_j) = 1$, which is equal to the constraint (\ref{OP0_c2}). Additional, the capacity bewteen the feed antenna node $v_{j}$ to the sink node is 1, we have $\sum_{i} f(u_k,v_j) \leq 1$, which is equal to the constraint (\ref{OP0_c3}). Hence, a maximum flow $f$ in the constructed graph corresponds to a feasible solution in problem (P1-1).
    \item \textbf{Optimality Guarantee:}
    The problem (P2-1) necessarily has an optimal solution, corresponding to the MCMF in the constructed network. Suppose there exists a better solution $\boldsymbol{\alpha}^*$ with higher lower bound of WSR than the MCMF solution $\boldsymbol{\alpha}^\circ$. Then, we can construct a flow $f^*$ with $f(u_{k},v_{j}) = \alpha^*_{s,n,k}, \forall j = (s-1)N_t+n$. According to the constraint (\ref{OP0_c2}), it would maintian equal total flow $|f^*| = K$ and have lower cost $\sum w_{k,j}f^*(u_{k},v_{j}) < \sum w_{k,j}f^\circ(u_{k},v_{j})$. This contradicts the minimality of $\boldsymbol{\alpha}^\circ$ in the MCMF solution. Therefore, the MCMF solution must be optimal for problem (P2-1).
\end{enumerate}
\end{proof}

\subsection{Beamforming Design of Metasurface}
With the obtained scheduling of feed antennas and TUs, we consider the beamforming design of metasurface. In order to maintain consistency in the optimization problem and ensure convergence, the lower bound of WSR is optimized, leading to the following optimization subproblem
\begin{subequations}
	\begin{align}
	(\text{P3}): \max_{\boldsymbol{\Theta}}~& \sum_{k=1}^{K}\beta_{k}\bar{R}_{k} \label{OP1_obj1}\\
	\text{s.t.} ~& \vert \left[\boldsymbol{\Theta}_{s,s} \right]_{l,l} \vert = 1, \forall s, \forall l. \label{OP1_c1}
	\end{align}
\end{subequations}
To reduce complexity and transform it into a more tractable form, we adopt a linear receive beamforming strategy by introducing receiver $\gamma_k$ for the $k$-th TU. Then, the mean square error (MSE) of the received signal at the $k$-th TU can be expressed as
\begin{align}
    \mathrm{MSE}_k &= \mathbb{E}\left[ \left| \gamma_k y_k^{\mathrm{com}} - x_k \right|^2 \right] \notag \\
    &=  |\gamma_k|^2 \left( \Lambda_k + \sigma_e^2 \sum_{s \in \mathcal{S}_k} \sum_{n=1}^{N_t} p_{s,n}^2 \rho_{s,k} \left\| \Theta_s \mathbf{g}_{s,n} \right\|^2 \right) \notag \\
	& + |\gamma_k \mu_k - 1|^2,
\end{align}
where
\begin{align}\label{mu_def}
    \mu_k = \sum_{s \in \mathcal{S}_k} \sum_{n=1}^{N_t} \alpha_{s,n,k} p_{s,n} \tilde{\mathbf{h}}_{s,k}^H \Theta_s \mathbf{g}_{s,n},
\end{align}
and
\begin{align}\label{lambda_def}
    \Lambda_k = \sum_{s \in \mathcal{S}_k} \sum_{n=1}^{N_t} p_{s,n}^2 \rho_{s,k}
\left|\tilde{\mathbf{h}}_{s,k}^H \Theta_s \mathbf{g}_{s,n} \right|^2 + \sigma_k^2.
\end{align}
By computing the derivative of $\text{MSE}_{k}$ with respect to $\gamma_{k}$, we can obtain the optimal MSE when $\gamma_{k} = \mu_{k}^{*}(\Lambda_k+ \sigma_e^2 \sum_{s \in \mathcal{S}_k} \sum_{n=1}^{N_t} p_{s,n}^2 \rho_{s,k} \left\| \Theta_s \mathbf{g}_{s,n} \right\|^2)^{-1}$, yielding the minimum MSE (MMSE) as $\text{MMSE}_{k} = (1 + \tilde{\Gamma}_{k})^{-1}$ \cite{CSIT}. Consequently, the original lower bound of WSR maximization problem (P3) can be equivalently transformed into
\begin{align}
\text{(P3-1)}:\min_{\boldsymbol{\Theta}, \boldsymbol{\gamma}} ~&\sum_{k=1}^{K}\beta_{k}\log_2(\text{MSE}_{k}) \label{min_mse}\\
\text{s.t.}~& \text{(\ref{OP1_c1})}. \notag
\end{align}
To further simplify the problem, we introduce auxiliary variables $\boldsymbol{\vartheta} = [\vartheta_{1},\cdots,\vartheta_{K}]$ following the WMMSE approach \cite{WMMSE method}, reformulating the problem as
\begin{subequations}
	\begin{align}
	\text{(P3-2)}:
	\min_{\boldsymbol{\Theta},\boldsymbol{\vartheta},\boldsymbol{\gamma}} & \sum_{k=1}^{K}\beta_{k}\left(\vartheta_{k}\text{MSE}_{k} - \log_2(\vartheta_{k}) \right) \label{wmmse_form}\\
	\text{s.t.} & \left|\boldsymbol{\varphi}_{s,m}\right| = 1, \quad \forall s,\forall m. \nonumber
	\end{align}
\end{subequations}
The optimal solution for $\vartheta_{k}$ is obtained by setting the first derivative of (\ref{wmmse_form}) to zero, yielding the closed-form expression $\vartheta_{k}^{\text{opt}} = \text{MSE}_{k}^{-1}$. Given that the optimization variable $\boldsymbol{\Theta}_{s}$ takes a matrix form and is coupled within the channel expression, it becomes challenging to handle under the unit modulus constraint. We expand $\text{MSE}_{k}$ in terms of $\boldsymbol{\varphi}_{s} = \text{vec}\{\boldsymbol{\Theta}_{s}\}$, which has a form of
\begin{align}
\mathrm{MSE}_{k} &= |\gamma_{k}|^2 \sum_{s\in\mathcal{S}_{k}} \sum_{n=1}^{N_t} p_{s,n}^2 \rho_{s,k}^2 \boldsymbol{\varphi}_{s} \Big( \tilde{\mathbf{f}}_{s,n,k} \tilde{\mathbf{f}}_{s,n,k}^{\text{H}} \notag \\ &\quad+ \sigma_e^2 \operatorname{diag} \!\left( \mathbf{g}_{s,n}\mathbf{g}_{s,n}^{\text{H}} \right) \Big) \boldsymbol{\varphi}_{s}^{\text{H}} \nonumber \\ &\quad - 2 \Re\Bigg\{ \gamma_{k} \sum_{s\in\mathcal{S}_{k}} \sum_{n=1}^{N_t} \alpha_{s,n,k} p_{s,n} \rho_{s,k} \boldsymbol{\varphi}_{s} \tilde{\mathbf{f}}_{s,n,k} \Bigg\} \nonumber \\ &\quad + |\gamma_{k}|^2 \sigma_{k}^{2} + 1, \label{eq:robust_global_mse}
\end{align}
where $\tilde{\mathbf{f}}_{s,n,k}=\operatorname{diag}(\mathbf{g}_{s,n}) \tilde{\mathbf{h}}_{s,k}$. After obtaining the optimal $\boldsymbol{\vartheta}$ and $\boldsymbol{\gamma}$, the problem can be reduced as
\begin{subequations}
	\begin{align}
	\text{(P3-3)}: \min_{\boldsymbol{\varphi}_{s}} & \sum_{s=1}^{S}\boldsymbol{\varphi}_{s} \boldsymbol{\Xi}_{s}\boldsymbol{\varphi}_{s}^{\text{H}} + 2\Re\left\{\sum_{s=1}^{S}\boldsymbol{\varphi}_{s}\boldsymbol{\eta}_{s}\right\} \notag\\
	& + \sum_{k=1}^{K}\beta_{k}\vartheta_{k}\left|\gamma_{k}\right|^2\sigma_{k}^2 + 1 \label{eq:reformulated_prob} \\
	\text{s.t.} & \left|\boldsymbol{\varphi}_{s,l}\right| = 1, \quad \forall s,\forall l,
	\end{align}
\end{subequations}
where
\begin{align}
\boldsymbol{\Xi}_s &= \sum_{k \in \mathcal{U}_s} \beta_k \vartheta_k |\gamma_k|^2
\sum_{n=1}^{N_t} p_{s,n}^2 \rho_{s,k}^2 \Big( \tilde{\mathbf{f}}_{s,n,k} \tilde{\mathbf{f}}_{s,n,k}^{H} \notag \\ &\quad + \sigma_e^2\operatorname{diag}\!\left(\mathbf{g}_{s,n}\mathbf{g}_{s,n}^{H}\right) \Big),
\label{eq:Xi_def}
\end{align}
and
\begin{align}
\boldsymbol{\eta}_s &= \sum_{k \in \mathcal{U}_s} \beta_k \vartheta_k \gamma_k \sum_{n=1}^{N_t} \alpha_{s,n,k} p_{s,n} \rho_{s,k} \tilde{\mathbf{f}}_{s,n,k}. \label{eq:eta_def}
\end{align}
However, the problem is still not convex due to the non-convex unit-modulus constraints. To address it, we employ the SDR technique with penalty function methods \cite{SDR method}. First, we introduce the augmented vector $\bar{\boldsymbol{\varphi}}_{s} = [\boldsymbol{\varphi}_{s},1] \in \mathbb{C}^{(L+1)\times1}$ and construct the extended matrix
\begin{equation}
\boldsymbol{\Psi}_{s} =
\begin{bmatrix}
\boldsymbol{\Xi}_{s} & \boldsymbol{\eta}_{s} \\
\boldsymbol{\eta}_{s}^{\text{H}} & 0
\end{bmatrix} \label{eq:Psi_matrix}.
\end{equation}
By this way, the objective function (\ref{eq:reformulated_prob}) can be rewritten as
\begin{equation}
\min_{\bar{\boldsymbol{\varphi}}_{s}} \sum_{s=1}^{S}\bar{\boldsymbol{\varphi}}_{s}\boldsymbol{\Psi}_{s}\bar{\boldsymbol{\varphi}}_{s}^{\text{H}} \label{eq:compact_form}.
\end{equation}
Through SDR technique, we define the rank-one positive semidefinite matrix $\boldsymbol{\Phi}_{s} = \bar{\boldsymbol{\varphi}}_{s}^{\text{H}}\bar{\boldsymbol{\varphi}}_{s}$, transforming the problem (P3-2) into
\begin{subequations}\label{SDR problem}
	\begin{align}
	\text{(P3-4)}:\min_{\boldsymbol{\Phi}_{s}} & \sum_{s=1}^{S}\text{tr}\{\boldsymbol{\Phi}_{s}\boldsymbol{\Psi}_{s}\} \label{eq:sdp_form} \\
	\text{s.t.} & [\boldsymbol{\Phi}_{s}]_{l,l} = 1, \quad \forall s,l,\label{P3-3st1}\\
	& \text{Rank}(\boldsymbol{\Phi}_{s}) = 1, \quad \forall s,\label{P3-3st2}\\
	& \boldsymbol{\Phi}_{s} \succeq 0, \forall s. \label{P3-3st3}
	\end{align}
\end{subequations}
Usually, after removing the rank-one constraint, this probelm can be efficiently solved using standard convex optimization tools, followed by a Gauss randomization procedure to extract feasible solutions satisfying the unit-modulus constraints \cite{SDR method}. However, this approach may lead to performance degradation, as the relaxation of the rank-one constraint and subsequent randomization procedure do not guarantee optimality. Therefore, we employ a penalty-based iterative approach to address the rank-one constraint \cite{Wangqi}, leading to the final relaxed formulation
\begin{align}
\text{(P3-5)}: \min_{\boldsymbol{\Phi}_{s}} & \sum_{s=1}^{S}\text{tr}\{\boldsymbol{\Phi}_{s}^{i}\boldsymbol{\Psi}_{s}\} + \rho^{i}\left|(\boldsymbol{\varphi}'_s)^{i}\boldsymbol{\Phi}_{s}^{i}(\boldsymbol{\varphi}'_{s})^{i} - L\right| \label{eq:final_form}\\
\text{s.t.} & (\text{\ref{P3-3st1}}),(\text{\ref{P3-3st3}}),
\end{align}
where $\rho^{i} > 0$ is the penalty coefficient and $\bar{\boldsymbol{\varphi}'}^{i}_{s}$ denotes the eigenvector corresponding to the largest eigenvalue of $\boldsymbol{\Phi}^{i-1}_{s}$ obtained after the $(i-1)$-th iteration. In particular, an appropriate penalty factor must be selected to balance convergence speed and constraint adherence. To address this, we employ a residual-based adaptive penalty function as
\begin{align} \label{adaptive penalty function}
	\rho^{i} = \rho^{i-1} + C \times \frac{\vert\boldsymbol{\varphi}'_s\boldsymbol{\Phi}_{s}\boldsymbol{\varphi}'_{s} - L - 1 \vert}{\epsilon},
\end{align}
where $\epsilon$ and $C$ are tolerance threshold and weighting coefficient, respectively. By iteratively solving the penalized problem and adaptively updating the penalty factor until convergence, we can obtain the optimal metasurface beamforming $\boldsymbol{\Theta}$ under the given receiver $\boldsymbol{\gamma}$ and auxiliary variables $\boldsymbol{\vartheta}$. Then, by further alternately updating the optimal receiver $\boldsymbol{\gamma}$, auxiliary variables $\boldsymbol{\vartheta}$, and $\boldsymbol{\Theta}$, we can ultimately obtain the optimal metasurface beamforming solution under the given scheduling.

Finally, by alternatingly optimizing scheduling $\boldsymbol{\alpha}$ and beamforming $\boldsymbol{\Theta}$, fixing one and optimizing the other, we can obtain the joint scheduling and beamforming results for LEO satellite constellation communication. The detailed steps are shown in Algorithm 2.

\begin{algorithm}
    \renewcommand{\algorithmicrequire}{\textbf{Input:}}
    \renewcommand{\algorithmicensure}{\textbf{Output:}}
    \caption{Joint Scheduling and Beamforming Design}
    \label{alg:2}
    \begin{algorithmic}[1]
        \REQUIRE $\hat{\mathbf{h}}_{s,k}$, $\mathbf{g}_{s,n}$, $N_t, L, K, \sigma_{k}^2$, $\mathcal{S}_{k}$;
        \ENSURE $\boldsymbol{\Theta}, \boldsymbol{\alpha}$;
        \STATE \textbf{Initialization:} Randomly generate $\boldsymbol{\Theta}^0 = \text{diag}(\boldsymbol{\varphi}^{(0)})$.
        \REPEAT
			\STATE Update scheduling of feed antenna and TUs $\boldsymbol{\alpha}^{j}$ using Algorithm 1;
			\STATE Set iteration counter $l = 0$ and $\text{OBJ}^0$ = 0 for the beamforming optimization of metasurface;
		\REPEAT
            \STATE Compute $\mu_k^{(l)}$ and $\Lambda_k^{(l)}$ using (\ref{mu_def}) and (\ref{lambda_def});
            \STATE Update $\gamma_k^{(l)} = (\mu_k^{(l)})/(\Lambda_k^{(l)})$;
            \STATE Update $\vartheta_k^{(l)} = 1/\text{MSE}_k^{(l)} $;
			\STATE Compute $\boldsymbol{\Xi}_s^{(l)}$ and $\boldsymbol{\eta}_s^{(l)}$ using (\ref{eq:Xi_def}) and (\ref{eq:eta_def});
			\STATE Construct extended matrix $\boldsymbol{\Psi}_s^{(l)}$ as (\ref{eq:Psi_matrix});
			\REPEAT
				\STATE Solve penalized SDP problem (\ref{eq:final_form}) for $\boldsymbol{\Phi}_s^{(i)}$;
				\STATE Perform eigen decomposition on $\boldsymbol{\Phi}_s^{(i)}$, extract principal eigenvector $\bar{\boldsymbol{\varphi}'}_s^{(i)}$;
				\STATE Update penalty coefficient $\rho^{i+1}$ according to (\ref{adaptive penalty function});
			\UNTIL{convergence.}
				\STATE Extract $\boldsymbol{\varphi}_s^{(l+1)}$ as first $L$ elements of $\sqrt{L+1}\bar{\boldsymbol{\varphi}'}_s^{(i)}$;
			\STATE  Calculate $\text{OBJ}^{l+1} = \sum_{k=1}^{K}\beta_{k}\bar{R}_{k}$ and calculate $\Delta = \text{OBJ}^{l+1}-\text{OBJ}^{l}$;
            \STATE $l=l+1$;
        \UNTIL{$\Delta < \epsilon $ or $l > L$.}
		\STATE Updata $\boldsymbol{\Theta}^{j}$ with $\boldsymbol{\varphi}_s^{(l+1)}$ from the step 17;
		\UNTIL{convergence.}
    \end{algorithmic}
\end{algorithm}

\subsection{Algorithm Analysis}
Herein, we analyze the proposed algorithm from the perspectives of convergence behavior and computational complexity.
\subsubsection{Convergence Analysis}
In Algorithm \ref{alg:2}, the optimal solutions of scheduling and beamforming are achieved by iteratively optimizing $\boldsymbol{\alpha}$ and $\boldsymbol{\Theta}$. Let $f(\boldsymbol{\alpha}^j,\boldsymbol{\Theta}^j)$ denote the lower bound of WSR with the obtained $\boldsymbol{\alpha}^{j}$ and $\boldsymbol{\Theta}^{j}$ after the $j$-th iteration. After obtaining $\boldsymbol{\alpha}^{j+1}$ by conducting Algorithm 1, we have
\begin{align}
	 f(\boldsymbol{\alpha}^{j+1}, \boldsymbol{\Theta}^j) \geq f(\boldsymbol{\alpha}^{j}, \boldsymbol{\Theta}^j).\label{inequ1}
\end{align}
Similarly, after obtaining $\boldsymbol{\Theta}^{j+1}$ by conducting steps 5-20 in Algorithm 2, we have
\begin{align}
	 f(\boldsymbol{\alpha}^{j+1}, \boldsymbol{\Theta}^{j+1}) \geq f(\boldsymbol{\alpha}^{j+1}, \boldsymbol{\Theta}^{j}).\label{inequ2}
\end{align}
Finally, with the inequalities (\ref{inequ1}) and (\ref{inequ2}), we can get
\begin{align}
	 f(\boldsymbol{\alpha}^{j+1}, \boldsymbol{\Theta}^{j+1}) \geq f(\boldsymbol{\alpha}^{j}, \boldsymbol{\Theta}^{j}).\label{inequ3}
\end{align}
The inequality (\ref{inequ3}) establishes that the lower bound of WSR exhibits a non-decreasing property across successive iterations. Furthermore, the WSR is inherently bounded above due to the finite transmit power constraints. Consequently, the algorithm is guaranteed to converge to a fixed point of the surrogate problem \cite{convergence}. It is note that this convergence guarantee does not extend to a stationary point of the original MINLP, which is consistent with standard AO-based approaches for mixed-integer non-convex optimization.

\subsubsection{Complexity Analysis}
Further, we analyze the complexity of the proposed Algorithm 2, which includes the MCMF-based Algorithm 1 and the beamforming design. For the MCMF-based Algorithm 1, it exhibits polynomial-time complexity. As illustrated in Fig. \ref{Constructed flow}, the network contains $K+SN_t$ nodes and $\sum_{k=1}^K |\mathcal{S}_k|N_t$ directed edges connecting TUs to their visible feed antennas. Each flow augmentation saturates at least one unit of capacity, and the maximum $K$ augmentations are required. The dominant complexity comes from the shortest path faster algorithm (SPFA) used in Step 2 of the Algorithm \ref{alg:unified_mcmf} to find the shortest path is $O((K+SN_t)(\sum_{k=1}^K |\mathcal{S}_k|N_t))$ per minimum-cost path computation. So the overall complexity of Algorithm 1 is $O(K(K+SN_t)(\sum_{k=1}^K |\mathcal{S}_k|N_t))$ \cite{SPFA}. The complexity of beamforming design primarily arises from solving the problem \text{(P3-4)}, which is imposed by $S(L+1)$ LMI constraints of size $1$ and $S$ LMI constraints of size $L + 1$. Based on the complexity analysis in \cite{Complexity analysis} for the SDP by the interior-point method, the complexity of beamforming per iteration is $\mathcal{O}(\sqrt{3S(L+1)}(S(L+1)^2(L+2)+S(L+1))\ln(1/\epsilon))$, where the $\epsilon > 0$ is the desired convergence accuracy. Overall, the complexity of the proposed Algorithm \ref{alg:2} is $\mathcal{O}(I_{O}((\sqrt{3S(L+1)}(S(L+1)^2(L+2)+S(L+1))\ln(1/\epsilon))+(K+SN_t)(\sum_{k=1}^K |\mathcal{S}_k|N_t)))$ with $I_{O}$ denotes the outer loop times of the alternating optimization of scheduling and beamforming. It is worth noting that the SDR-based metasurface beamforming subproblem, which exhibits cubic complexity $\mathcal{O}(L^3)$, is not intended for real-time onboard execution. Instead, this optimization is performed offline or on a slow time scale corresponding to satellite geometry updates or significant changes in user distribution. In contrast, the online onboard processing is limited to the polynomial-complexity MCMF-based scheduling and lightweight phase selection operations, which are well suited for space-grade processors.

\section{Simulation Result}
In this section, we evaluate the performance of the proposed joing scheduling and beamforming design algorithm for metasurface antenna enabled LEO satellite constellation communication. Without loss of generatity, we adopt the Walker delta constellation for simulation, the detailed configuration of which is listed in Table \ref{Table1} \cite{Walker Delta constellation}. Then, we select three adjacent satellites to form a group providing services for the ground TUs. It is assumed that the relative geometry and ISL connectivity within the satellite cluster remain quasi-static over each optimization window, since orbital motion is smooth and predictable \cite{Satellite_geometry}. As a result, topology changes and handovers occur on a slower time scale than resource allocation and beamforming optimization, which enables the proposed framework to operate under a locally stable cluster topology during each optimization period. In particular, the initial position of the LEO satellites are given by $\mathbf{q}^{\text{sat}}_{1} = 10^5 \times [-4.0426,60.887,5.5235]$, $\mathbf{q}^{\text{sat}}_{2} = 10^5 \times [-1.5338,64.153,2.0957]$ and $\mathbf{q}^{\text{sat}}_{3} = 10^5 \times [-1.6248,68.527,5.9641]$ within the Earth-centered Earth-fixed (ECEF) coordinate system, respectively. The coverage angle of satellite can be computed as $\eta = \arccos(R_e/(R_e + h_O)\cos \theta_{\text{min}}-\theta_{\text{min}})$ with $h_O$ and $\theta_{\min}$ being the altitude of satellite orbit and minimum elevation angle of TUs, respectively. The total number of TUs is set to $K = 30$, which are divided into three equal groups. Each group of 10 TUs is randomly and uniformly distributed within the coverage area of one of the three satellites. Then, the relative angle bewteen the $k$-th TU position angle and the $s$-th satellite position angle can be computed as $  \varphi_{s,k} = \arccos(\frac{\mathbf{q}_{k}^{u} \cdot \mathbf{q}_{s}^{\text{sat}}}{\Vert \mathbf{q}_{k}^{u} \Vert \cdot \Vert \mathbf{q}_{s}^{\text{sat}} \Vert})$ where $\mathbf{q}_{k}^{\text{u}}$ represents the coordinate of the $k$-th TU. By this way, it is easy to obtain the set of covering TUs $\mathcal{U}_{s}$ of the $s$-th satellite, i.e., $\varphi_{s,k}\leq \eta$ for the TUs belong to set $\mathcal{U}_{s}$. Unless other stated, the parameters of LEO SATCOM system are listed in Table \ref{Table2}.

\begin{table}
    \centering
    \caption{Parameters of LEO Satellite Constellation}
    \label{Table1}
    \begin{tabular}{|c|c|} \hline
    \textbf{Parameter} & \textbf{Value} \\ \hline
    Orbital altitude $h_{O}$& 550km \\ \hline
    Number of orbital planes $N^{\#}$ & 36  \\ \hline
    Number of LEO satellites per orbital plane $P^{\#}$ & 22 \\ \hline
    Orbital inclination & $53^\circ$ \\ \hline
    Minimum elevation angle of TUs $\theta_{\min}$ & $10^\circ$ \\ \hline
    Phase factor & 1 \\ \hline
    \end{tabular}
\end{table}

\begin{table}[!ht]
    \centering
    \caption{Parameter Setup of LEO Satellite Constellation System}
    \label{Table2}
    \begin{tabular}{|c|c|} \hline
    \textbf{Parameter} & \textbf{Value} \\ \hline
    Carrier frequency $f_c$ &	30 GHz  \\ \hline
    Spectrum bandwidth $B$ & 25 MHz  \\ \hline
    Number of LEO satellites in a group $S$ & 3 \\ \hline
    Number of feed antennas $N_t$ & 13 \\ \hline
    Transmit power budget of feed antenna $P_{s,n}$ & 30 dBm \\ \hline
    Total number of TUs $s$ $K_{s}$ & 30\\ \hline
    Number of elements of metasurface $L_x\times L_y$ & 20 $\times $ 20\\ \hline
	Area of each meta-atom $A_{t}$ &  25 mm$^2$\\ \hline
	Inter element spacing of metasurface &  5 mm\\ \hline
    Noise temperature $T$ & 290 K \\ \hline
    Boltzmann constant $\kappa$ & $1.38\times10^{-23} J/m$\\ \hline
    Rain attenuation mean $\mu_r$ & -2.6 dB \\ \hline
    Rain attenuation variance $\sigma_{r}^2$ & 1.63 dB \\ \hline
	Angular deviations $\delta_{s,k,p},\xi_{s,k,p}$ & $\mathcal{N}(0, 0.01) ^\circ$ \\ \hline
    Maximum iteration number $T^{\max}$ & 30\\ \hline
	Channel error variance$\sigma_e^2$ & 1/SNR\\ \hline
    \end{tabular}
\end{table}

\begin{figure}[!ht]
	\centering
	\includegraphics[width=3.4in]{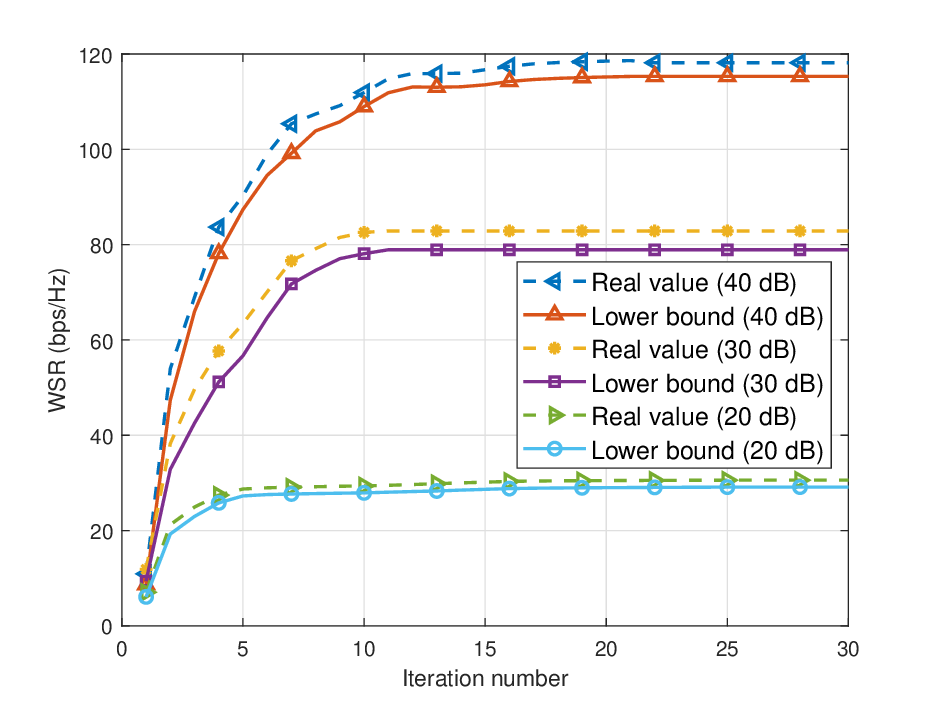}
	\caption{Convergence behaviour of the proposed Algorithm \ref{alg:2}.}
	\label{convergence}
\end{figure}
First, Fig. \ref{convergence} demonstrates the convergence behaviour of the proposed Algorithm \ref{alg:2}. It can be observed that the optimized WSR lower bound closely matches the actual WSR, and both curves exhibit nearly identical convergence behavior, converging within 30 iterations under different feed antenna transmit power settings. In particular, due to the slow variation of large-scale channel characteristics in LEO systems and the robustness provided by the CSI error model, the proposed optimization is executed on a coarse time scale and does not require per-slot real-time updates. The results confirm the rapid convergence of the proposed joint scheduling and beamforming algorithm, demonstrating its applicability in real-world LEO satellite communications.

\begin{figure}[!ht]
	\centering
	\includegraphics[width=3.4in]{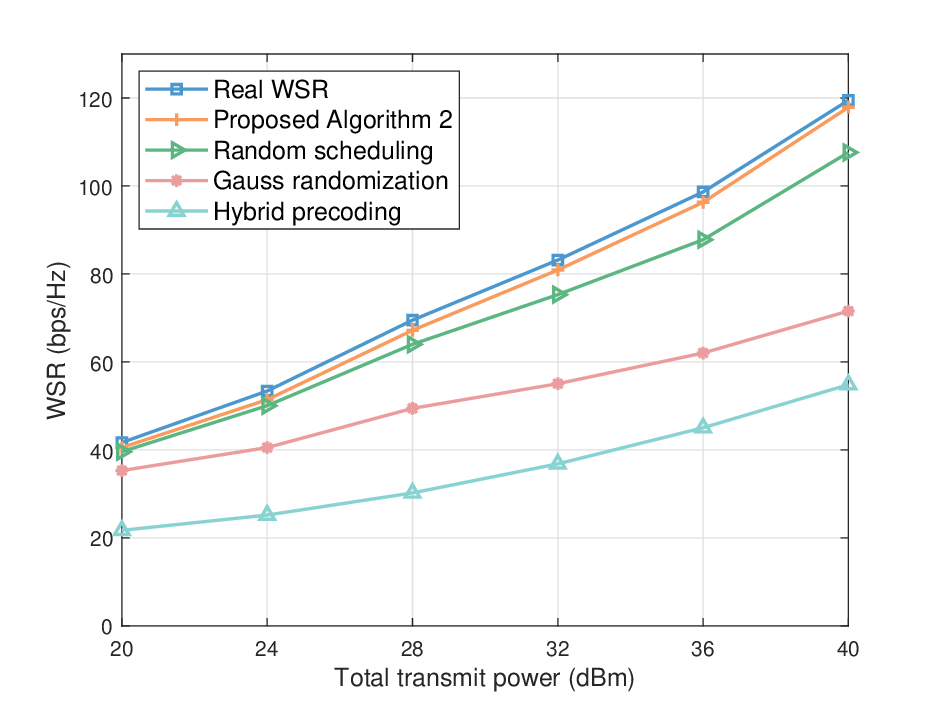}
	\caption{WSR versus the transmit power of antenna.}
	\label{WSRversusPower}
\end{figure}
Next, we evaluate the WSR performance versus the total transmit power of the feed antenna, comparing our proposed Algorithm \ref{alg:2} against three benchmark schemes: (i) random scheduling of feed antennas and TUs, (ii) metasurface optimization using Gaussian randomization to handle rank-one constraints, and (iii) hybrid precoding system with $N_t = 13$ RF chains and 10 $\times$ 10 uniform planar array (UPA) antennas. The results demonstrate that the proposed scheduling algorithm achieves substantial gains over random scheduling, while significantly outperforming the Gaussian randomization scheme for metasurface optimization. This improvement is expected because, unlike the Gauss randomization scheme, the proposed algorithm optimizes the metasurface without relaxing the constraints. Notably, the proposed algorithm consistently outperforms the hybrid beamforming benchmark under identical channel and power constraints, and its performance advantage stems from the additional spatial degrees of freedom introduced by the metasurface architecture. Furthermore, this gain can be progressively enhanced by increasing the number of low-cost passive meta-atoms, enabling scalable performance improvement without the need for additional RF chains or active transceiver hardware.

\begin{figure}[!ht]
	\centering
	\includegraphics[width=3.4in]{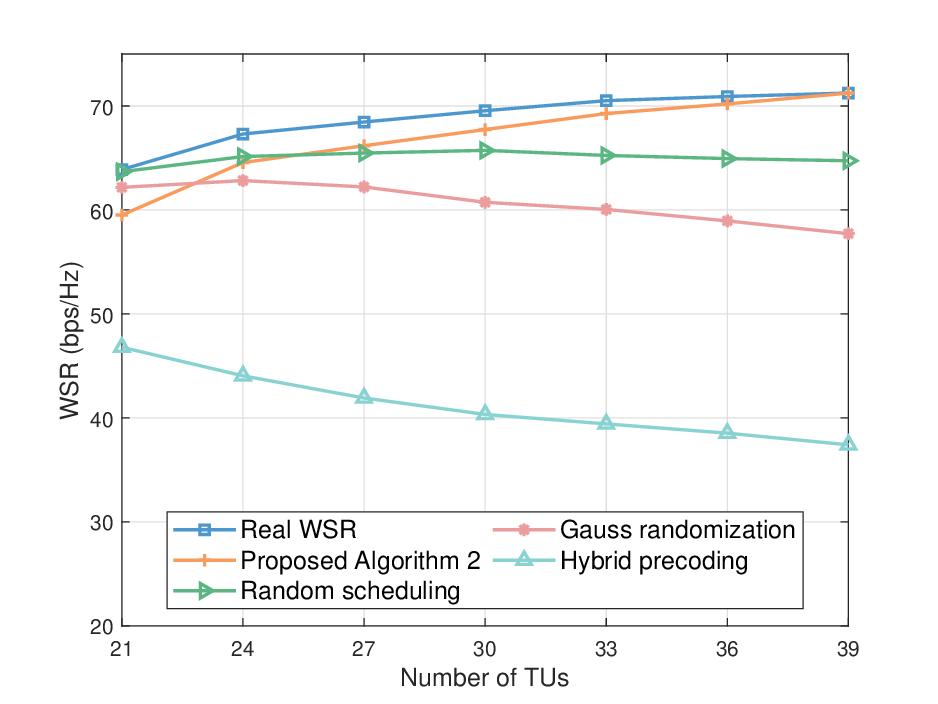}
	\caption{WSR versus the number of TUs.}
	\label{WSRversusTUs}
\end{figure}

In Fig. \ref{WSRversusTUs}, we show the WSR performance under different numbers of TUs. The results reveal that as the number of TUs increases, the gap between the optimized lower bound of WSR and the real WSR gradually narrows, eventually converging to identical values. This convergence phenomenon occurs because a larger number of active TUs leads to full utilization of the available feed antennas, making the interference pattern in the optimized lower bound approximation increasingly representative of the real system conditions. Consequently, the theoretical lower bound becomes an accurate predictor of the achievable WSR when all feed antennas are actively engaged in transmission. Additionally, the proposed algorithm consistently demonstrates performance gains over benchmark schemes across various numbers of TUs. Notably, it maintains superior spectral efficiency even in fully loaded scenarios where the number of TUs equals the number of feed antennas, outperforming conventional hybrid precoding. This advantage stems from the algorithm's joint optimization of scheduling and beamforming, which more effectively mitigates inter-user interference. Furthermore, the interference suppression enabled by the metasurface architecture contributes to its robust performance, particularly in interference-dominated regimes.

\begin{figure}[!ht]
	\centering
	\includegraphics[width=3.4in]{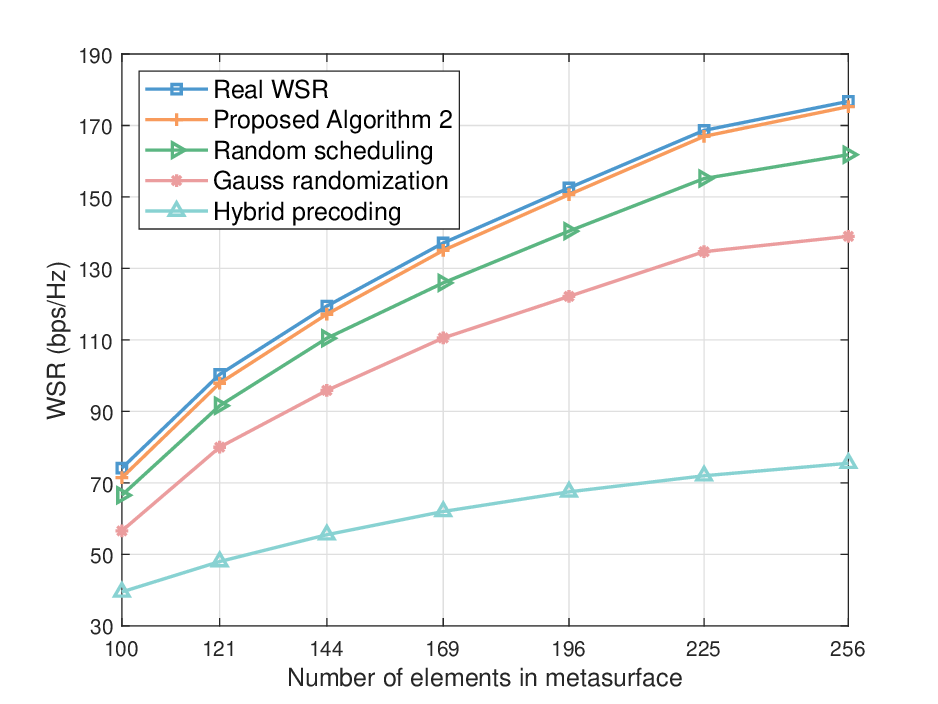}
	\caption{WSR versus the number of elements in metasurface.}
	\label{WSRversusMx}
\end{figure}
Fig. \ref{WSRversusMx} illustrates the WSR performance versus the number of metasurface elements. The results demonstrate that all schemes exhibit improved WSR as the number of elements increases, while the rate of improvement slows down gradually. This nonlinear growth pattern occurs because when the metasurface contains few elements, adding more elements provides dual benefits. The enhanced signal power stems from improved beamforming gain, while superior interference suppression is achieved through more precise spatial filtering. However, once the number of elements reaches a sufficient level, typically around the spatial degrees of freedom required by the system, the marginal improvement primarily comes from additional power gain, while the interference management capability approaches its theoretical limit. This explains why the WSR growth rate gradually slows as the metasurface size increases beyond this critical point. For the considered user density and channel conditions, the results in Fig.~\ref{WSRversusMx} indicate that the performance saturation appears when the number of metasurface elements becomes sufficiently larger than the number of simultaneously active TUs. This provides a practical design guideline for hardware deployment, suggesting that excessive over-provisioning of metasurface elements may lead to limited performance benefits while significantly increasing hardware complexity and system cost.

\begin{figure}[!ht]
	\centering
	\includegraphics[width=3.4in]{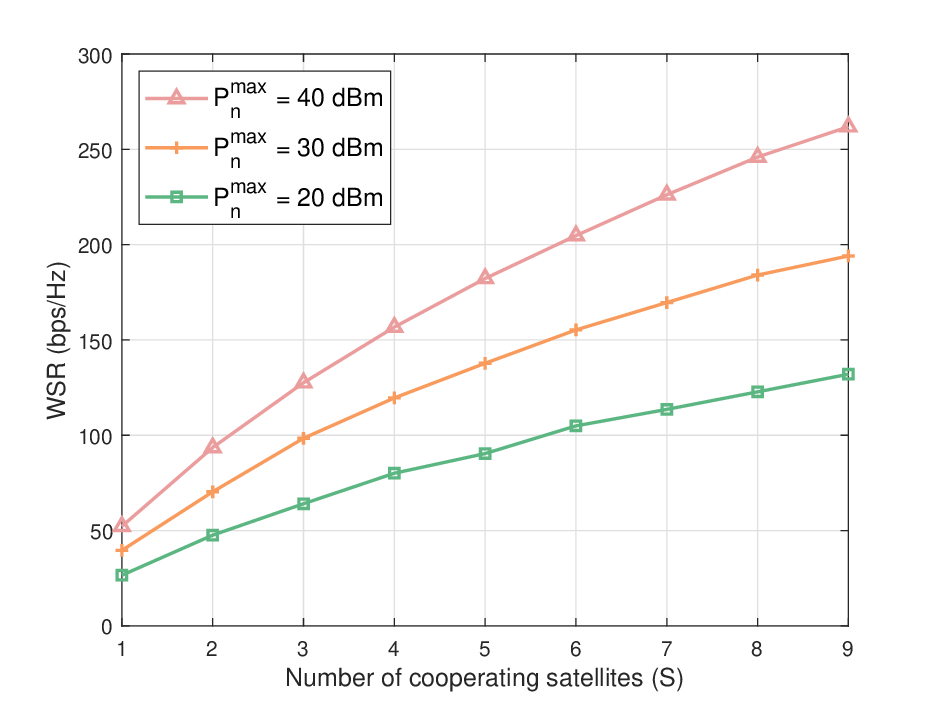}
	\caption{WSR versus the number of cooperating satellites.}
	\label{WSRversussatellites}
\end{figure}

Then, Fig.~\ref{WSRversussatellites} illustrates the impact of the number of cooperating satellites on the achievable WSR under different maximum transmit power constraints. As the cooperative cluster expands, the WSR increases in a quasi-linear manner. This behavior reflects the combined effect of enhanced spatial degrees of freedom and distributed transmission capability enabled by multi-satellite cooperation. By leveraging coordinated beamforming across multiple satellites, co-channel interference is progressively suppressed, leading to more efficient spatial resource utilization and improved spectral efficiency. Meanwhile, the distributed nature of the cooperative architecture enables performance gains through both power aggregation and spatial diversity rather than relying solely on power scaling. Moreover, the results confirm the scalability of the proposed framework for large cooperative clusters. Even when the number of cooperating satellites increases to nine, the system maintains stable and consistent performance improvements.

\begin{figure}[!ht]
	\centering
	\includegraphics[width=3.4in]{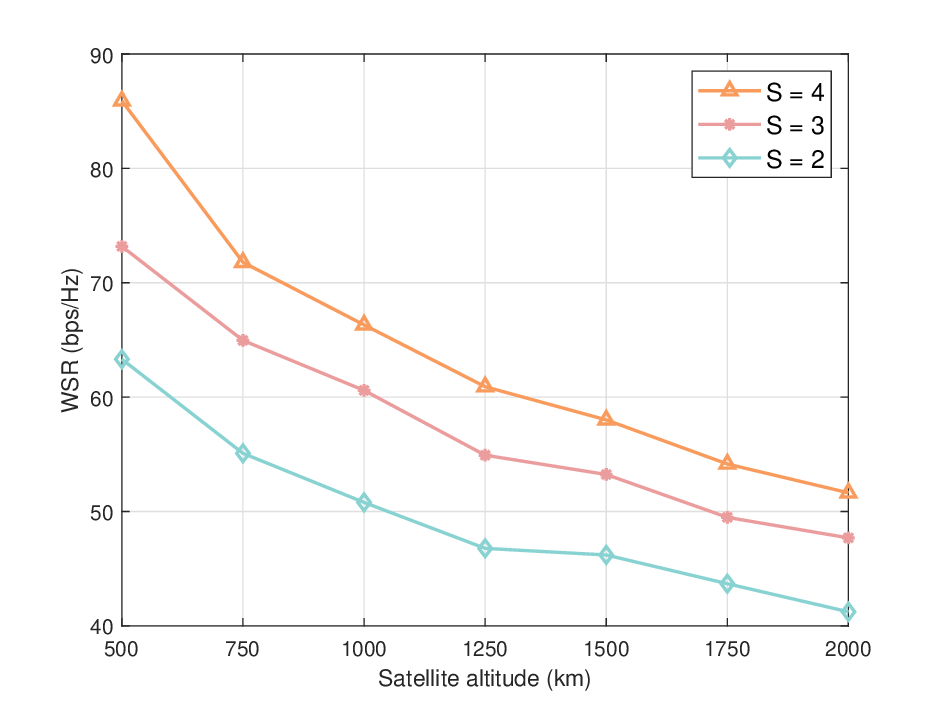}
	\caption{WSR versus the altitude of satellites.}
	\label{WSRversusAltitude}
\end{figure}
Fig. \ref{WSRversusAltitude} illustrates the impact of satellite orbit altitude on WSR performance under varying numbers of cooperating satellites. As observed, the WSR deteriorates with increasing orbital altitude, primarily due to higher propagation path loss. This finding aligns with existing research advocating for very low Earth orbit (VLEO) satellite constellations, which mitigate path loss by operating at significantly reduced altitudes \cite{VLEO}. Furthermore, the results demonstrate that a greater number of cooperating satellites enhances the overall system performance, highlighting the benefits of multi-satellite collaboration in mitigating channel impairments.

\begin{figure}[!ht]
	\centering
	\includegraphics[width=3.4in]{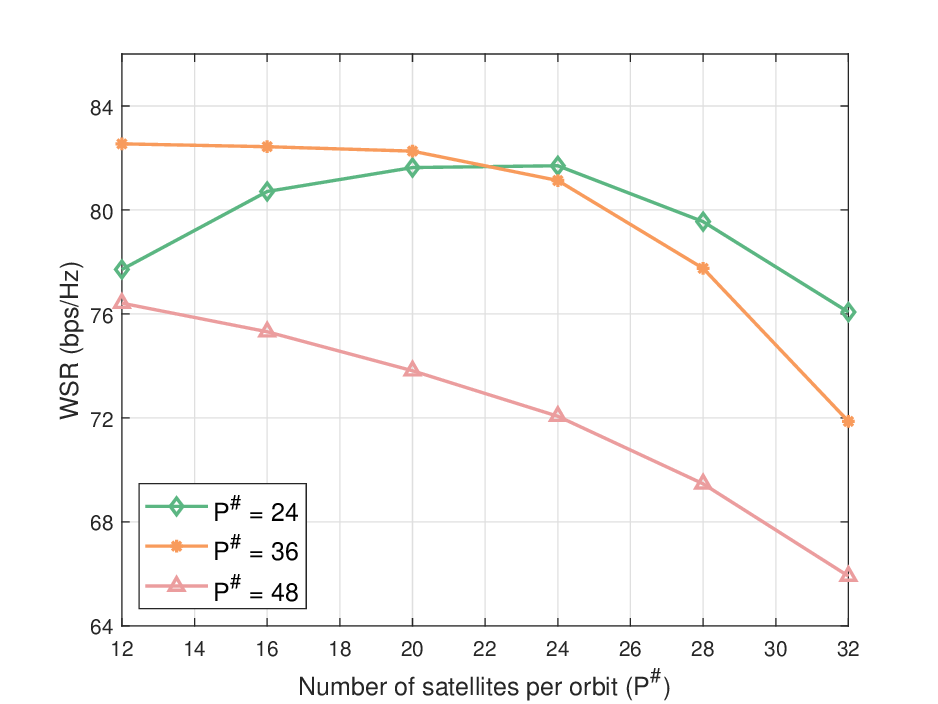}
	\caption{WSR versus the density of the LEO satellite constellation.}
	\label{WSRversusDensity}
\end{figure}
Finally, Fig. \ref{WSRversusDensity} demonstrates the variation in system performance in terms of WSR with different satellite densities. It can be observed that when the number of orbits is $N^{\#} = 24$, the WSR initially increases but then decreases as the number of satellites per orbit grows. For a larger number of satellite orbits ($N^{\#} = 36$ and $N^{\#} = 48$), the WSR shows a more significant decline with increasing satellite density. This behavior stems from competing effects that additional satellites improve user access to the nearest visible satellite, and thus enhance link quality at lower densities. However, beyond a certain point, further increment in satellite density leads to stronger inter-satellite interference. In addition, higher density increases deployment costs and collision risks. These results highlight the critical need to balance spectral efficiency with practical considerations of cost and operational safety in LEO constellation design.

\section{Conclusion}
This paper established a novel MA-enabled LEO satellite constellations communication framework. By shifting beamforming operations from the digital domain to EM wavefront manipulation, we realized hardware-efficient signal processing while mitigating interference in satellite constellation networks. The proposed joint scheduling and beamforming algorithm, which synthesizes graph-based scheduling design and convex-relaxed phase shift design, resolved the complex optimization design problem with tractable computational overhead. Extensive simulations confirmed substantial improvements compared with the existing benchmarks. Future extensions of this framework may incorporate learning-based surrogate models to enable ultra-low-latency metasurface beamforming. Furthermore, the proposed metasurface-enabled LEO satellite framework exhibits strong potential for integration with energy-efficient software-defined 5G/6G multimedia IoV systems, facilitating its incorporation into integrated space--air--ground networks and further enhancing scalability, real-time adaptability, and practical applicability. In particular, LEO satellite constellations can provide seamless coverage extension and high-throughput connectivity for vehicular services in remote or infrastructure-scarce environments, while the leader-satellite architecture naturally aligns with software-defined networking paradigms for cross-layer resource orchestration. The inherent energy efficiency of metasurface-based beamforming further supports sustainable and reliable multimedia vehicular communications, enabling the development of future integrated space--air--ground IoV networks.

\end{document}